\documentclass[twoside,11pt]{article}

\usepackage{jmlr2e}

\usepackage{amsmath}
\usepackage{color}
\usepackage{float} 
\usepackage[linesnumbered,ruled,lined]{algorithm2e} 
\usepackage{mathrsfs} 
\usepackage{hyperref}
\usepackage[capitalize]{cleveref}
\usepackage{xcolor}
\usepackage{xspace}
\usepackage{booktabs} 
\usepackage{multirow}
\usepackage{graphicx}
\usepackage{subfig}
\usepackage{caption}

\makeatletter
\DeclareRobustCommand\onedot{\futurelet\@let@token\@onedot}
\def\@onedot{\ifx\@let@token.\else.\null\fi\xspace}
 
\def\ie{\emph{i.e}\onedot}

\newcommand{\Gstruct}{ $G_{l_1 + l_2}$\xspace}
\newcommand{\Mstruct}{ Q^{\textsf{S}}}
\newcommand{\Grandom}{ $G_{l_1 + N}$\xspace}
\newcommand{\Mrandom}{ Q^{\textsf{R}}}

\jmlrheading{1}{}{}{}{}{}{Xiaodong Xin, Kun He, Jialu Bao,Bart Selman and John E. Hopcroft}

\ShortHeadings{Structure Amplification on Multi-layer SBM}{Xin, He, Bao, Selman and Hopcroft} 
\firstpageno{1}

\begin{document}

\title{Structure Amplification on Multi-layer Stochastic Block Models}

\author{\name Xiaodong Xin\thanks{The first three authors contribute equally.} \email M201973348@hust.edu.cn 
\AND
\name Kun He~\thanks{Corresponding author.} \email brooklet60@hust.edu.cn\\ 
\addr School of Computer Science and Technology\\
Huazhong University of Science and Technology\\
Wuhan 430074, China
\AND
\name Jialu Bao \email jialu@cs.wisc.edu \\
\addr Department of Computer Sciences\\
Madison, WI 50706, USA
\AND
\name Bart Selman \email selman@cs.cornell.edu \\
\AND
\name John E. Hopcroft \email jeh@cs.cornell.edu\\
\addr Department of Computer Science\\
Cornell University\\
Ithaca, NY 14853
}

\editor{}

\maketitle

\begin{abstract}
Much of the complexity of social, biological, and engineered systems arises from a network of complex interactions connecting many basic components. Network analysis tools have been successful at uncovering latent structure termed communities in such networks. However, some of the most interesting structure can be difficult to uncover because it is obscured by the more dominant structure. Our previous work proposes a general structure amplification technique called HICODE that uncovers many layers of functional hidden structure in complex networks. HICODE incrementally weakens dominant structure through randomization allowing the hidden functionality to emerge, and uncovers these hidden structure in real-world networks that previous methods rarely uncover. 
In this work, we conduct a comprehensive and systematic theoretical analysis on the hidden community structure. 
In what follows, we define multi-layer stochastic block model, and provide theoretical support using the model on why the existence of hidden structure will make the detection of dominant structure harder compared with equivalent random noise. We then provide theoretical proofs that the iterative reducing methods could help promote the uncovering of hidden structure as well as boosting the detection quality of dominant structure.
\end{abstract}


\section{Introduction}\label{section(Introduction)}
The complexity of many real-world systems emerges from the interactions of large numbers of interconnected components. The network underlying these interactions can uncover the structure or functionally related parts. The general approach of network analysis is to identify groups of network nodes that have a higher density of connections within a group than across groups. Such sets of highly interconnected nodes are referred to as “communities”~\citep{2002Community}, inspired by the notion of communities of individuals in the social network (e.g., friends or colleagues). The notion captures the intuition that having many “within the group” connections suggests a common or related functional role of the nodes, for instance the gene groups in Protein-Protein Interaction (PPI) networks~\citep{2010IsoBase}.

Over the recent decades, a rich set of tools have been developed to find communities in large networks~\citep{lancichinetti2009community,xie2013overlapping}. Most of these works are for the global community structure, including disjoint communities that partition the nodes~\citep{2008Fast, Rosvall2008Infomaps}, or overlapping  community  structure~\citep{palla2005uncovering, Ahn2010Link, Andrea2011OSLOM}. With the rapid growth of network scale, there is also an increasing trend of shifting the attention to the local community detection~\citep{clauset2005finding, kloster2014heat, LaarhovenM16JMLR, he2019krylov}.

All the above methods have shown promising results in analyzing real-world networks. Which algorithm performs best is often domain dependent. Overall, these methods perform the best when the various communities in the network are well-defined and of comparable strength, as defined by the ratio of the density of in-community edges versus the density of cross-community edges. However, real-world data can have different kinds of structure corresponding to different types of functional relationships. Some of the most interesting structure can be obscured by the more dominant structure and thus easily overlooked. For example, we may want to track the spread of a new, potentially lethal flu virus in social networks, but the signal may be hidden under the structure arising from the more frequent online social interactions. Similarly, the scientific literature co-author network in emerging interdisciplinary areas are often overshadowed by the far more dominant co-author relationships within traditional research areas. For another example, biologists are interested in identifying gene groups in Protein-Protein Interaction (PPI) networks \citep{2010IsoBase}. The clearest groupings are often already known to the field, and the more valuable discovery would be the hidden, less obvious groupings.
Such type of structure is very difficult for current algorithms to find because the strong primary structure drives the grouping of nodes (or “clustering”) towards the primary structure.

To resolve this issue and to properly handle communities of different strengths, which are induced by different structure in complex networks, our previous works~\citep{He15corr,he18} propose a new concept called the hidden community. A community is called \textit{the hidden structure} if most of the members also belong to other stronger communities, as evaluated by some community scoring function such as modularity~\citep{2002Community}. Note that this conception of ``hidden" differs from the normal task of community detection that detects the latent communities from observed interactions (\ie edges)~\footnote{Researchers in the literature sometimes also call the community to be mined the hidden community, which differs to our conception of ``hidden community". For consistency, we call all the communities to be mined the latent communities.}.

Our previous work~\citep{he18} also proposes a structure amplification technique called HICODE (HIdden COmmunity DEtection) for mining dominant communities as well as hidden communities formed by \textit{layers}, where each layer is formed by a set of disjoint communities or overlapping communities found by typical community detection methods. \cite{he18} further provide empirical demonstrations using synthetic networks as well as real-world networks that exhibit multiple layers of natural community structure.

The concept of hidden structure has attracted increasing attention since then.
\cite{he2019contextual} use the information between different structures for graph clustering. \cite{li2020single} are inspired to iteratively remove the reflection of a single image captured through a glass surface, in which they regard the transmission as the strong and dominant structure while the reflection the weak and hidden structure. \cite{nath2019detecting} focus on the intrinsic community and extend the hidden community detection to dynamic networks. 
\cite{gong2018finding} propose multi-granularity community detection method (MGCD) based on network embedding to detect the hidden communities. \cite{salz2019hidden} use HICODE on more real-world networks to verify the effectiveness of HICODE.

However, \cite{he18} did not provide theoretical support for the proposed HICODE. 
We provide theoretical analysis on HICODE using two-layer stochastic models in a preliminary version\footnote{This manuscript is a significant extension of our conference version~\citep{bao20}.}, and in this work we further develop a general theoretical guarantee of HICODE on the multiple layers of stochastic block models. Specifically, we prove that the modularity of target layer increases in the process of using HICODE, and explain why HICODE significantly improves the detection from both theoretical and simulation perspectives. Some intuitive examples are also provided to facilitate the understanding.

Meanwhile, we observe that if we organize the adjacency matrix by dominant structure, the hidden structure appears to be random noise; however, it is not truly random and we call it the structured noise. This not only raises a question on what is the difference between structured noise and random noise but also inspires us to explore whether the existence of the structured noise hinders detecting the dominant structure more. In this work, we also investigate these significant questions from a specific angle.

The structure of the manuscript is as follows. Section \ref{section(Preliminary)} briefly introduces HICODE and formally proposes the stochastic block models. In Section \ref{section(random)} we explain why the existence of hidden structure makes it hard for algorithm to even detect the dominant community structure, and why the detection of hidden structure is necessary. We start from the three-layer stochastic block models to demonstrate the effectiveness of HICODE in Section \ref{section(Theoretical analysis on three-layer Stochastic Block Model)} and expand the theoretical analysis to multi-layer stochastic block models in Section \ref{section(Theoretical analysis on Multi-layer Stochastic Block Model)}. We then conduct simulation to indicate the great performance of the algorithm in Section \ref{(section)Simulation of HICODE} and finally conclude the paper with some outlook of future works.

\section{Preliminary}\label{section(Preliminary)}

For preliminary, we first introduce the measuring metric of modularity for a single community as well as for a layer of communities. Then, we show how HICODE works to detect the hidden  structure by structure amplification. Also, we define multi-layer stochastic block model as an abstraction of real-world networks, and illustrate the running results of HICODE on a four-layer stochastic block model. In the end, we categorize the edges into different sets according to the different intersection of layers they belongs, which will be used in the followup theoretical analysis.

\subsection{Modularity}

\cite{2002Community} propose an important metric called \emph{modularity} to evaluate the strength for a set of communities that partitions the nodes of a network. A higher modularity indicates the denser the internal edges and the sparser the external edges of a community partition. The metric is widely used by community detection algorithms, such as the Louvain method~\citep{2008Fast}, which finds the modularity-maximizing partition.


\begin{definition}[Modularity of a Community]\label{Modularity of a community and a partition/layer}
Given a graph $G $ $= (V, E)$ with a total of $e$ edges and multiple layers of communities, where each layer of communities partitions all nodes in the graph, for a community $i$ in layer $l$, let $e_{ll}^i$ denote $i$'s internal edges, and $e_{lout}^i$ denote the number of edges having exactly one endpoint in community $i$. Let $d_l^i$ be the total degree of the nodes in community $i$ ($d_l^i = 2 e_{ll}^{i} + e_{lout}^i$). Then the modularity of community $i$ in layer $l$ is:
\begin{equation*}
Q_l^i = \frac{e_{ll}^{i}}{e} - \left(\frac{d_{l}^{i}}{2e} \right)^2.
\end{equation*}
\end{definition}
In this work, we will work on layers of communities often.
A single layer of communities corresponds to a partitioning or covering of the nodes in the network.
Thus, the modularity of a layer follows from the conventional definition of the modularity of a partition.

\begin{definition}[Modularity of a Layer]\label{def: layer modularity}
For a given graph $G$, for a layer $l$, say $l$ partitions all the nodes into disjoint communities $\{1, \dots, N\}$, then the modularity of this layer is $Q_l = \sum_{i=1}^N Q_l^i$, where $Q_l^i$ is the modularity of a community on the $G$ defined above.
\end{definition}

Intuitively, the higher fraction of internal edges a community has among all edges, the more connected its members are, and thus a higher modularity of the community.

\subsection{The HICODE Algorithm}

HICODE~\citep{he18} is a structure amplification approach that repeatedly modifies the network under consideration. To find communities in the modified network, it uses a given community finding method, referred to as the ``base algorithm''.
We show the key procedure of the amplification of base algorithm $\mathbf{X}$ in Figure~\ref{figure:flowchart}.

The approach assumes the network consists of distinct layers of communities.
 The amplification technique starts with an initial estimate of the communities using the base algorithm \textbf{X} in $k$ distinct layers (the identification stage). HICODE subsequently improves the community structure in each layer in an iterative manner (the refinement stage). The basic idea for improving the approximation of a particular layer $l$ is to first weaken all structure resulted from communities from the other layers, and then recompute the community structure using the base algorithm \textbf{X}. The weakened network will reveal the structure of $l$-th layer more clearly than the original one. Such structure amplification is repeated for each layer in the network.

\begin{figure}[htb]
\center{\includegraphics[width=8cm] {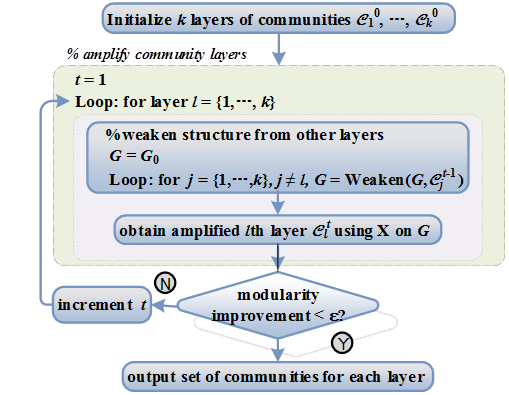}}
\caption{\label{figure:flowchart} Amplification procedure using the base algorithm $\mathbf{X}$. $G_0$ is the input network. $\mathcal{C}_j^t$ gives the set of communities in layer $j$ after iteration $t$.}
\end{figure}

There are three reducing methods for the Weaken($G,\mathcal{C}$) function.
\begin{itemize}
	\item $\bf{RemoveEdge:}$ RemoveEdge weakens the structure of a detected layer $l$ by removing internal edges of all communities of layer $l$, that is, all connections that could come from the communities in $l$. 

	\item $\bf{ReduceEdge:}$ Instead of removing all internal edges of $l$, ReduceEdge probabilistically removes some internal edges of each community $i$ in layer $l$ so that the edge density in community $i$, denoted $\widehat{p_{l}^{i}}$, matches with the edge density in the background, denoted $\widehat{q_{l}^{i}}$.
	$\widehat{p_{l}^{i}}$ and $\widehat{q_{l}^{i}}$ are approximated by the local information as follows:
\begin{align}
	\label{defn_qi}
\widehat{p_{l}^{i}} = \frac{e_{ll}^{i}}{\frac{1}{2}s_{l}^{i}(s_{l}^{i}-1)},   \qquad \qquad \widehat{q_{l}^{i}} = \frac{e_{lout}^{i}}{s_{l}^{i}(n-s_{l}^{i})},
\end{align}
where $e_{ll}^{i}$ and $e_{lout}^{i}$ represent the internal edges and outgoing edges of community $i$ in layer $l$, $s_{l}^{i}$ and $n$ represent the size of community $i$ in $l$ and the number of nodes. Thus, in order to reduce $\widehat{p_{l}^{i}}$ to $\widehat{q_{l}^{i}}$ we can keep each internal edge with probability $q_{l}^{i}=\widehat{q_{l}^{i}}/\widehat{p_{l}^{i}}$, indicating that we randomly remove each internal edge of community $i$ with probability $1-q_{l}^{i}$. 

\item $\bf{ReduceWeight:}$ ReduceWeight can be considered as the counterpart and de-random\\-ization of ReduceEdge on the weighted graph and it also needs to approximate the retention probability $q_{l}^{i}$. ReduceWeight weakens structures in layer $l$ by multiplying the weight of all edges in community $i$ by $q_{l}^{i}$.
\end{itemize}

\subsection{Stochastic Block Model}
To reflect the fact that real-world networks often exhibit patterns in the high level but exhibit uncertainties in the edge level, 
researchers model them through stochastic block model ~\citep{abbe2017community,deng2021strong}.
In particular, we use \emph{multi-layer stochastic block model} to model networks with multiple layers of communities.
Roughly, a single-layer stochastic block model network consists of blocks (i.e., communities), each of which is a Erd\H{o}s-R\'{e}nyi graph $G(n, n_l, p_l)$; 
a multi-layer stochastic block model network is an union of multiple independent single-layer stochastic block model networks on a shared set of nodes.

\begin{definition}[Multi-layer stochastic block model] \label{(definition)Multi-layer stochastic block model}
	A multi-layer stochastic block model network $G(n, n_1, n_2,$ $...,n_L, p_1, p_2,..., p_L )$, in which $n, n_{1}, n_{2},..., n_{L} \in N^{+}$,
	has $n$ nodes and $L$ layers.
For $l = 1, 2,..., L$, layer $l$ of $G$ consists of $n_l$ planted communities of size $s_{l} = \frac{n}{n_{l}}$, and pairs of nodes within the same community form edges with probability $p_{l}$.

Communities in different layers are grouped independently, that is,
for any $k$ layers $\{l_{1}, l_{2},...,$ $ l_{k}\}$ $(2 \leq k \leq L)$, if we pick a community $c_i$ from each layer $l_i$, then the expected number of nodes in the intersection of $c_1, \dots, c_k$ is always $r_{l_{1}l_{2}...l_{k}} = \frac{n}{n_{l_{1}}n_{l_{2}}...n_{l_{k}}}$.
\end{definition}

The multi-layer stochastic block model could represent an ideal case where there is no noise and all outgoing edges of one layer are internal edges of some other layers, and each community of layer $l$ is expected to have $ \frac{1}{2} s_{l} \cdot (s_{l}-1) \cdot p_l$ internal edges.
In fact, the number of communities and generation probability of an edge in real networks tend not to be too small, we will assume that $n_{l} \geq 4$ and $p_{l} \in [0.05,1]$ for any layer $l$ in the following.
Besides, to make sure every intersection block has at least two nodes, we require $n \geq 2\prod_{l=1}^{L} n_l$.
At last, as communities in the same layer are expected to have equal sizes and are both independent of any communities in other layers,
we denote the expected number of internal edges (resp. outgoing edges) in any community of layer $l$ as $e_{ll}$ (resp. $e_{lout}$).
Note that uniform random noise could be added through setting an additional layer with $n_l = 1$ and $p_l$ be the noise generation probability.

\begin{lemma}\label{jialu2020}%
For layer $l$ in the stochastic block model, if the layer weakening method (one of RemoveEdge, ReduceEdge, ReduceWeight) reduces a bigger percentage of outgoing edges than internal edges, \ie the expected number of internal and outgoing edges after weakening, $e'_{ll}$ and $e'_{l out}$, satisfy $\frac{e'_{l out}}{e_{l out}} < \frac{e'_{l l}}{e_{l l}}$, then the modularity of layer $l$ increases after the weakening method.
\end{lemma}

\begin{proof}
Since all communities in a layer have the same number of nodes,
internal edges, and outgoing edges, in expectation, we have $e=\frac{n_{l}d_{l}}{2}$.
\begin{align}
Q_l &= n_l \cdot \left[ \frac{e_{ll}}{e} - \left(\frac{d_l}{2e} \right)^2 \right] = n_l \cdot \left[ 2\frac{e_{l l}}{n_l\cdot d_l} - \left (\frac{d_l}{n_l \cdot d_l } \right)^2 \right] \notag\\
&= 2 \cdot \frac{e_{ll}}{ d_l} - \frac{1}{n_l}  
= 1 - \frac{e_{lout}}{e_{lout} + 2e_{l l}} - \frac{1}{n_l}. \label{(equation)(jialu2020)}
\end{align}

If  $\frac{e'_{l out}}{e_{l out}} < \frac{e'_{l l}}{e_{l l}}$ after weakening,
then $\frac{e'_{lout} + 2e' _{l l}}{e'_{lout}} > \frac{e_{lout} + 2e _{l l}}{e_{lout}}$,
and thus the modularity of the original layer on the weakened graph, $Q_l'$ will
exceed $Q_l$.
Intuitively, if more percentage of outgoing edges of layer $l$ is reduced than internal edges, there are less percentage of cross-community edges, and then the modularity of layer $l$ increases.
\end{proof}

\subsection{Illustration of HICODE Results on a Stochastic Block Model}

As an illustrative example on how HICODE runs, consider
a synthetic network built by the four-layer stochastic block model plus a layer
of noise.  There are 240 nodes, and the four layers have 10, 8, 6, 5
communities of size roughly 24, 30, 40 and 48 respectively. The inner edge probability for
different layers are 0.30, 0.22, 0.17 and 0.10, so that the modularity measures
are descending and deeper layers are weaker, with modularity scores of 0.231,
0.210, 0.197 and 0.130.  Additional edges are added to the network with
probability 0.005 as the background noise.  For example, this synthetic network can represent
a network of academic papers where papers can be clustered by topic
(physics, chemistry, or mathematics, \dots), by type of article (survey, expository, research, tool paper, \dots), by type of publication (journal, book, conference proceedings, \dots), and by the native language of authors (English, Spanish, Chinese, \dots).

Figure \ref{figure:synL4_results} illustrates snapshots at different times of the four detected layers during the amplification process.
After 60 iterations, the approach has almost fully uncovered all community structure in all four layers. Each community found by HICODE (using Louvain as the base algorithm) has a near perfect match with one of the embedded communities, and vice versa. Other methods find less than half the structure for any of the four layers. Some methods just yield the whole network as the only community.

\begin{figure}[htb!]
\center{
\includegraphics[width=0.9\textwidth]  {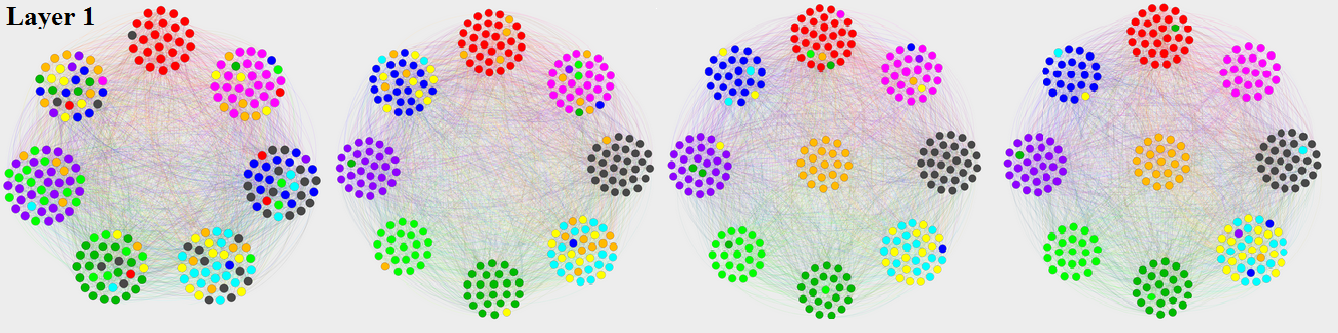}\\
\includegraphics[width=0.9\textwidth]  {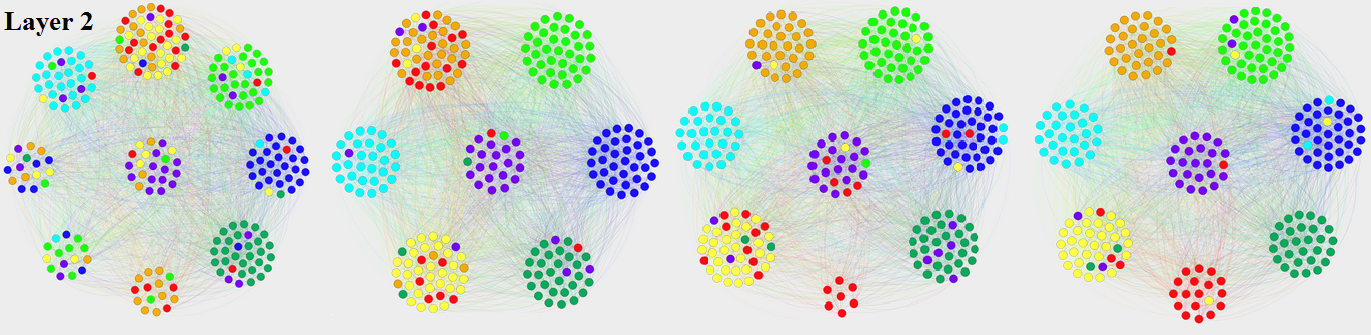}\\
\includegraphics[width=0.9\textwidth]  {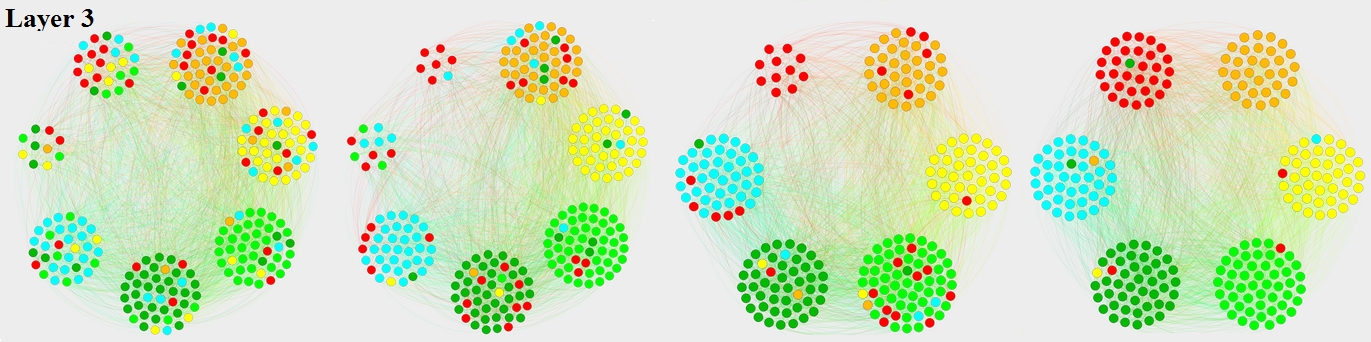}\\
\includegraphics[width=0.9\textwidth]  {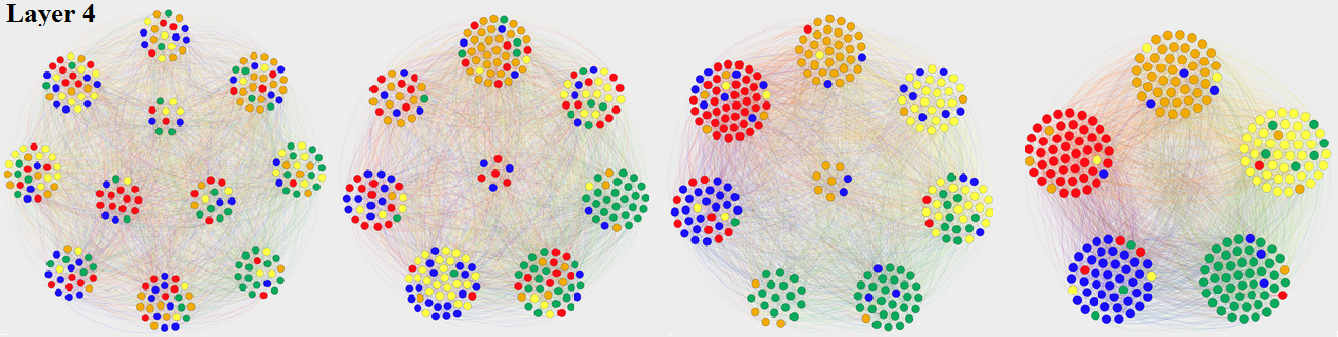}\\
\text{~~~~$t= 0$~~~~~~~~~~~~~~~~~~$t= 20$~~~~~~~~~~~~~~~~~~~$t = 30$~~~~~~~~~~~~~~~~~~~~$t = 60$~~}
}
\caption{Illustration of the structure amplification for the four layers of the SBM network. For each layer, we color nodes from each ground truth community with a distinct color and group nodes in one estimated community into one cluster.
The leftest panel ($t = 0$) indicates the communities initially detected. In Layer 1, the base algorithm finds 7 communities, with only one of them clearly corresponding to a ground truth community (red nodes). Similarly, after weakening the structure of the initial three layers, at $t = 0$, the base algorithm finds 11 communities in Layer 4, but none that corresponds to a ground truth community (all communities have a mixture of colors). The following snapshots ($t = 20, 30, 60$) show how the four layers are gradually uncovered as the coloring of nodes in each community becomes more uniform over time. At $t$ = 60, almost all communities are nearly monochromatic, and therefore correspond well to the ground truth communities.}
\label{figure:synL4_results}
\end{figure}

Table \ref{(tab)Acc_comp} shows the accuracy comparison of HICODE and several
typical community detection methods:
Louvain~\citep{2008Fast} (the base algorithm), OSLOM~\citep{lancichinetti2011finding}, LC~\citep{Ahn2010Link}, GCE~\citep{lee2010detecting}, DEMON~\citep{coscia2012demon}, CFinder~\citep{gregory2009finding} and CONGA~\citep{palla2005uncovering}.
We will compare their results via Jaccard Recall, Precision and F1 score, which measure the accuracy of the detected communities compared to annotated communities (“ground truth”) in the network. Given two communities, $A$ and $B$, the Jaccard similarity is given by the ratio between the size of the intersection of $A$ and $B$ and the size of the union of $A$ and $B$, that is, $|A \cap B| / |A \cup B|$. For each annotated community, its individual Jaccard Recall is the greatest Jaccard similarity over all detected communities. For each detected community, its individual Jaccard Precision is the greatest Jaccard similarity over all annotated (ground truth) communities. The Jaccard Recall $R$ is the average of individual Recalls among annotated communities, weighted by the fraction of nodes in each annotated  community. The Jaccard Precision $P$ is defined similarly. The Jaccard F1 score $F$ is the harmonic mean of $R$ and $P$.

\begin{table*}[htb]
\centering
\resizebox{\textwidth}{35mm}{
\begin{tabular}{lr|cccccc|c|cccc}
\toprule
& & \multicolumn{6}{c|}{\textbf{Overlapping Algorithms}} & \textbf{Base} & \multicolumn{4}{c}{\textbf{HICODE(Louvain)}} \\
\textbf{SynL4} & & \textbf{OSLOM} & \textbf{LC} & \textbf{GCE} & \textbf{DEMON} & \textbf{CFinder} & \textbf{GONGA} & \textbf{Louvain} & $L_1$ & $L_2$ & $L_3$ & $L_4$ \\
\midrule
$PL_1$ & $R$ & 0.28 & 0.10 & 0.10 & 0.16 & 0.23 & 0.23 & 0.42 & \textbf{0.92} & 0.12 & 0.13 & 0.12\\
    &    $P$ & 0.31 & 0.12 & 0.12 & 0.16 & 0.14 & 0.19 & 0.50 & \textbf{0.92} & 0.12 & 0.13 & 0.14\\
    &    $F$ & 0.29 & 0.11 & 0.11 & 0.16 & 0.17 & 0.21 & 0.46 & \textbf{0.92} & 0.12 & 0.13 & 0.13\\
\midrule
$PL_2$ & $R$ & 0.23 & 0.13 & 0.13 & 0.17 & 0.17 & 0.22 & 0.21 & 0.11 & \textbf{0.93} & 0.13 & 0.13\\
    &    $P$ & 0.25 & 0.15 & 0.15 & 0.18 & 0.12 & 0.19 & 0.22 & 0.11 &\textbf{0.93} & 0.15 & 0.14\\
    &    $F$ & 0.24 & 0.14 & 0.14 & 0.18 & 0.14 & 0.20 & 0.21 & 0.11 &\textbf{0.93} & 0.14 & 0.13\\
\midrule
$PL_3$ & $R$ & 0.30 & 0.17 & 0.17 & 0.20 & 0.18 & 0.25 & 0.21 & 0.13 & 0.13 & \textbf{0.90} & 0.14\\
    &    $P$ & 0.29 & 0.20 & 0.20 & 0.20 & 0.13 & 0.21 & 0.19 & 0.13 & 0.12 & \textbf{0.91} & 0.14\\
    &    $F$ & 0.30 & 0.19 & 0.19 & 0.20 & 0.15 & 0.23 & 0.20 & 0.13 & 0.13 & \textbf{0.90} & 0.14\\
\midrule
$PL_4$ & $R$ & 0.19 & 0.20 & 0.20 & 0.22 & 0.20 & 0.20 & 0.16 & 0.14 & 0.15 & 0.15 & \textbf{0.85}\\
    &    $P$ & 0.19 & 0.22 & 0.22 & 0.22 & 0.12 & 0.16 & 0.15 & 0.12 & 0.13 & 0.15 & \textbf{0.85} \\
    &    $F$ & 0.19 & 0.21 & 0.21 & 0.22 & 0.15 & 0.18 & 0.15 & 0.13 & 0.14 & 0.15 & \textbf{0.85}\\
\midrule
\# Comm & & 6 & 1 & 1 & 9 & 71 & 29 & 7 & 10 & 8 & 6 & 5 \\
\bottomrule
\end{tabular}}
\caption{Accuracy  comparison  for  HICODE  and  several  typical  community  detection methods. The rows $P$, $R$, $F$ indicate Jaccard Precision, Recall, and F1 scores respectively.
For a layer $k$, $PL_k$ indicates the planted communities, $L_k$ indicates the detected communities by HICODE.}\label{(tab)Acc_comp}
\end{table*}

Intuitively speaking, each layer of communities corresponds to structure of a particular type of global functionality.
The strength of community structure can vary significantly between layers but the iterative amplification strategy focuses on individual layers and can therefore detect layers of significantly different strength. This allows us to uncover very weak structure in deeper layers that other methods cannot detect.

\subsection{Edge Set for the Intersection of Layers}

For a multi-layer stochastic block model graph, we categorize edges that belong to the intersection of internal blocks of different layers.
For instance in a two-layer stochastic block model $G(100,4,5,0.5,0.4)$,
there are three edge sets, as illustrated in Figure \ref{fig:AMo2sbm} by different coloring of edges in the adjacency matrix. Blue (resp. red) elements represent edges only internal to layer 1 (resp. layer 2), while green elements represent edges internal to both layers. They have different generation probabilities, $0.5, 0.4,$ and $1 - (1 - 0.5)(1 - 0.4) = 0.7$, respectively.


\begin{figure}[htbp]  
    \centering
    \subfloat[Ordered by blocks in layer 1]
    {
        \begin{minipage}[t]{0.5\textwidth}
            \centering
            \includegraphics[width=0.6\textwidth]{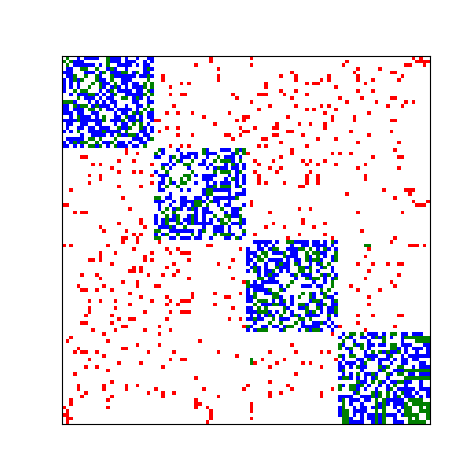}
        \end{minipage}%
    }
    \subfloat[Ordered by blocks in layer 2]
    {
        \begin{minipage}[t]{0.5\textwidth}
            \centering
            \includegraphics[width=0.6\textwidth]{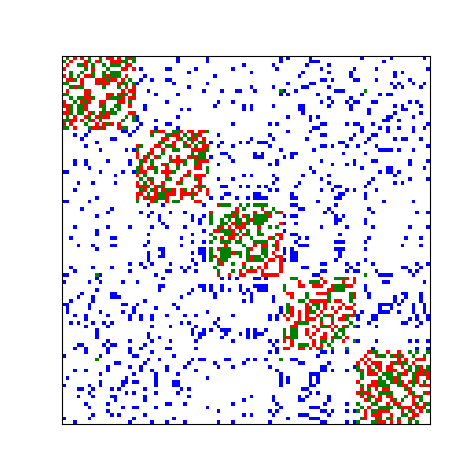}
        \end{minipage}
    }%

    \caption{Three edge sets for a two layer SBM (illustrated in different colors).}
    \label{fig:AMo2sbm}
\end{figure}

\begin{definition}[Edge Set]
\label{Set of different edges in the graph}
For a stochastic block graph model $G(n,$ $ n_1, n_2$ $,...,n_L, p_1, p_2,...,$ $ p_L)$, we divide edges into different sets according to the common layers they belong to: set $S_{l_{1}l_{2},...,{l_{k}}}$ represents edges that are internal to the intersection of blocks of layers $l_{1}$, $l_{2}$,...,$l_{k}$.
\end{definition}


For $G(n, n_1, n_2,$ $...,n_L, p_1, p_2,..., p_L)$,
due to the independent generation probability of different layers, we use $p_{l_{1}l_{2}..l_{k}} = 1 -\prod_{\{l_{1}l_{2}..l_{k}\}}(1-p_{l_{i}})$ to denote the generation probability of an edge in set $S_{l_{1}l_{2},...,{l_{k}}}$. We also let $|\overline{S_{l_{1}l_{2},...,{l_{k}}}}|$ denote the expected number of edges.

\section{Structured Noise versus Random Noise}\label{section(random)}

In some real world networks there exist a layer of dominant structure and several layers of secondary or hidden structures~\citep{he18}. In this section, using the synthetic network of stochastic block model, we show that these hidden structures, which appear to be random noise but indeed are not independently random, can create problems for typical community detection algorithms on real-world data.


Specifically, for the adjacent matrix of a two-layer stochastic block model, when we group nodes using partitions in the dominant layer, edges in the hidden layer are scattered like noises, as illustrated in Fig. \ref{fig:AMo2sbm} (a), which we refer to as the structured noise.
We will explain why having structured noise makes the discovery of the dominant communities harder than just having the random noise.
The high-level intuition is that:
most state-of-the-art community detection algorithms employ strategies to find partitions with high modularity (or similar measures on connectivity);
with the presence of structured noise, the partition that matches the dominant layer will exhibit lower modularity, and there will exist other partitions that locally maximize the modularity, making it harder for the detection of dominant layer.

More precisely, consider a two-layer stochastic block model $G(n,n_1, n_2,$ $ p_1, p_2)$, denoted as shorthand \Gstruct, and a one-layer stochastic block model with noise, which could be generated by $G(n, n_1, 1, p_1, p_n)$, denoted as shorthand \Grandom. To control the variables in comparison, we set $p_n$ such that the expected number of edges generated as noise in \Grandom is equal to the expected number of edges generated for layer 2 in \Gstruct, that is, $p_n = \frac{s_{2}-1}{n-1} \cdot p_{2}$ where $s_2=\frac{n}{n_2}$.

\begin{figure}[htb]
\center{\includegraphics[width=15cm]  {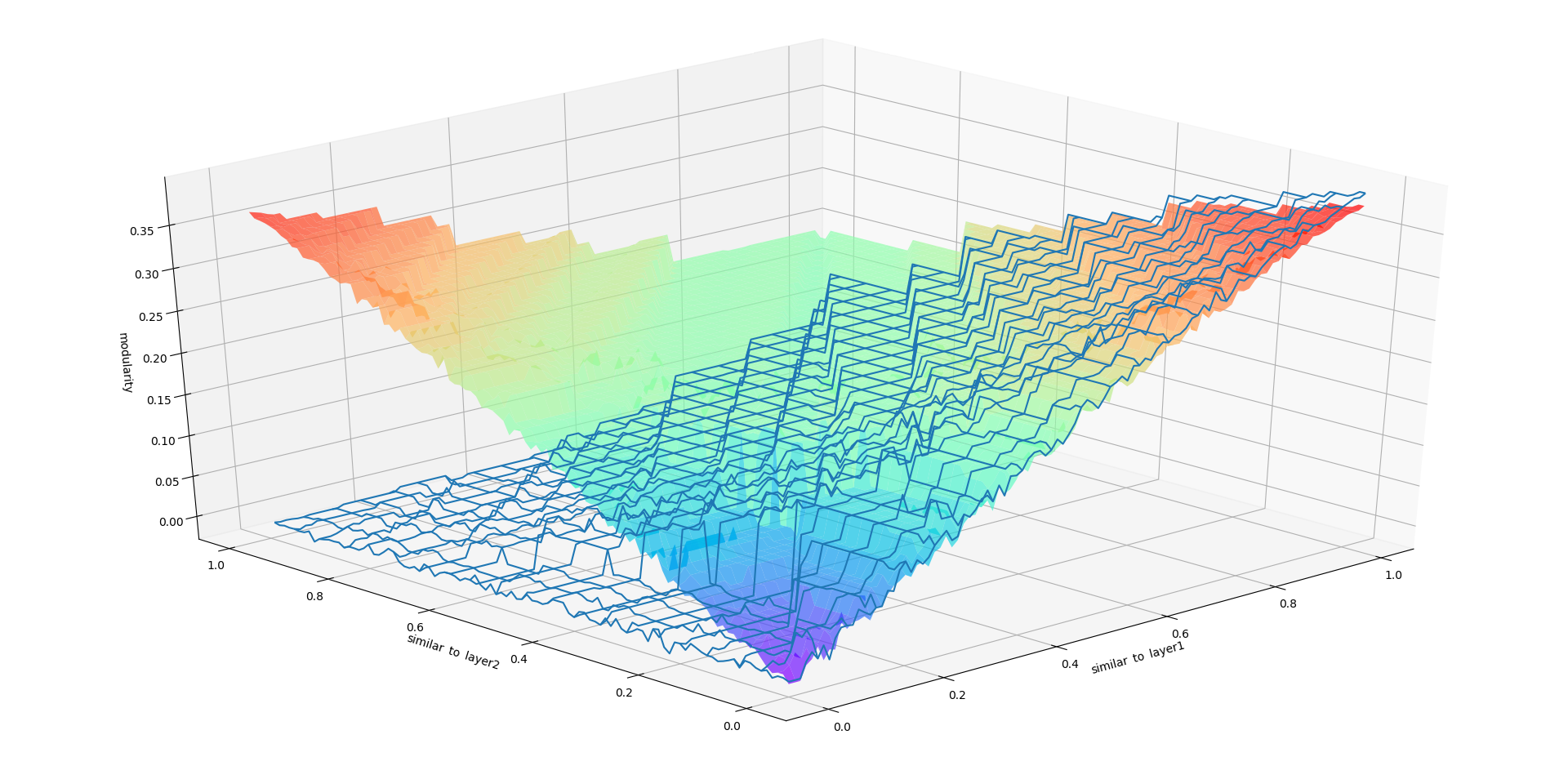}}
\caption{Illustration of the modularity scores of different partitions for \Gstruct (rainbow surface) versus \Grandom (blue wireframe).} 
\label{(figure)modularity_of_two_model}
\end{figure}

To pictorially show how the structured noise creates more difficulties than purely random noise,
we simulate with $n, n_1, n_2, p_1, p_2$ equal to 200, 4, 5, 0.10, 0.11 respectively and visualize the modularity of various partitions on \Gstruct and \Grandom.
In~\cref{(figure)modularity_of_two_model},
we use a partition's similarity with the ground truth community layers $l_1$ and $l_2$ as its $x$ and $y$-coordinates and take its modularity score as its $z$-coordinate, and the similarity is measured by the normalized mutual information score \citep{danon2005comparing}.
The rainbow surface shows the partitions' modularity in \Gstruct and the blue wireframe shows their modularity in \Grandom.
The right-end point in the plane (coordinates $(1,0)$) represents the dominant layer $l_1$, and the figure shows that its modularity is slightly lower in \Gstruct than in \Grandom.
More significantly, the modularity of the left-end point (coordinates $(0,1)$), which represents the hidden layer $l_2$, is much higher in \Gstruct than in \Grandom. And more importantly, it forms another local maximum in the rainbow surface; the existence of two local peaks indicates a harder optimization task in the community detection.
In the subsequent subsections, we will analyze the expected modularity rigorously for the general $n, n_1, n_2, p_1, p_2$.


\subsection{Structured Noise Reduces More Quality}
Although the structured noise looks similar to purely random noise,
we can prove that they affect modularity differently.
Denote the modularity of $l_1$ on the \Gstruct and \Grandom
as $\Mstruct$ and $\Mrandom$. We will show that
$\Mstruct_{l_1} < \Mrandom_{l_1}$.
By Lemma~\ref{jialu2020}, we have 
\begin{align}\label{(equation)jialu2020_in_why_use_hicode}
\Mstruct_{l_1}< \Mrandom_{l_1} \iff
\frac{e_{11}^{\textsf{S}}}{e_{1out}^{\textsf{S}}} < \frac{e_{11}^{\textsf{R}}}{e_{1out}^{\textsf{R}}},
\end{align}
where $e_{11}^{\textsf{S}}$ ($e_{11}^{\textsf{R}}$) and $e_{1out}^{\textsf{S}}$ ($e_{1out}^{\textsf{R}}$) represent the  number of internal edges and outgoing edges 
in model \Gstruct (\Grandom).
To prove the right inequality, we use a few more definitions and lemmas.

Recall in~\cref{Set of different edges in the graph}, we can divide
the edges in \Gstruct into sets $S_{1}, S_{2}$, and $S_{12}$,
and divide the edges in \Grandom into sets $S_{1N}$ and $S_N$ (in \Grandom, there is no node pair that is in $l_1$ but not the noise layer $N$. ).
Then, we will use these edge sets to calculate
the expected number of internal edges and outgoing edges of $l_1$ in
different models.
Denote the expected size of a set $S$ as $|\overline{S}|$.

\begin{lemma}\label{(lemma)ell_and_elout_2_layer_model}
The numbers of internal edges and outgoing edges for $l_1$ in \Gstruct and \Grandom are respectively:
\begin{align*}
&e^{\textsf{S}}_{11} = \frac{1}{ n_1}(|\overline{S_{12}}| + |\overline{S_1}|),
&e^{\textsf{S}}_{1out} = \frac{2}{n_1} |\overline{S_2}|,\\
&e^{\textsf{R}}_{11} = \frac{1}{n_1}|\overline{S_{1N}}|,
&e^{\textsf{R}}_{1out} = \frac{2}{n_1} |\overline{S_N}|.\\
\end{align*}
\end{lemma}


\begin{proof}
	Immediate from the fact that each outgoing edge is counted as outgoing for two communities, and that each internal edge is internal for exact one community.
%
%
\end{proof}

The expected size of different edge sets are functions of parameters in the stochastic block model.

\begin{lemma}[Expected number of edges in different edge sets on \Gstruct and \Grandom]\label{(lemma)(why use hicode)The maximum number of edges in different set}
The value of $|\overline{S_{1}}|$,$|\overline{S_{2}}|$,$|\overline{S_{12}}|$,$|\overline{S_{N}}|$ ,and $|\overline{S_{1N}}|$ are respectively:
\begin{align*}
&|\overline{S_{1}}| = \frac{1}{2}n(s_{1}-r_{12})p_{1}, \qquad |\overline{S_{2}}| = \frac{1}{2}n(s_{2}-r_{12})p_{2},\\
&|\overline{S_{12}}| =  \frac{1}{2} n(r_{12}-1)p_{12}, \qquad
|\overline{S_{N}}| = \frac{1}{2} n(n-s_{1})p_{n},\\
&|\overline{S_{1N}}| =\frac{1}{2}n(s_{1}-1)p_{1n}.\\
\end{align*}
\end{lemma}

\begin{proof}
In \Gstruct, every community in $l_1$ intersects with every community in $l_2$, so there are $n_1 \cdot n_2$ intersection blocks. Every intersection block contains $r_{12}$ nodes, so there are a total of $n_1 n_2
\cdot \frac{1}{2}r_{12}(r_{12}-1)$ node pairs in the intersection
with the edge generation probability $p_{12}$.
Therefore, $|\overline{S_{12}}| = n_1 n_2\cdot \frac{1}{2}r_{12}(r_{12}-1)p_{12} = \frac{1}{2} n(r_{12}-1)p_{12}$.

Each community in $l_1$ has $\frac{1}{2}s_{1}(s_{1}-1)$ node pairs, and thus,
$n_1$ such communities have $n_1 \cdot \frac{1}{2}s_{1}(s_{1}-1)$ node pairs in total. Excluding
$\frac{1}{2} n(r_{12}-1) $ node pairs that are also internal to $l_2$, there
are $n_1 \cdot \frac{1}{2}s_{1}(s_{1}-1) -  \frac{1}{2} n(r_{12}-1)$ node
pairs that are only internal to $l_1$, which implies
$|\overline{S_{1}}| = \frac{1}{2}n(s_{1}-r_{12})p_{1}$.
Similarly, $|\overline{S_{2}}| = \frac{1}{2}n(s_{2}-r_{12})p_{2}$.

In \Grandom, we can think of the noise layer as one block containing all $n$ nodes. Then, analogous to the calculation of $|\overline{S_2}$ and $|\overline{S_{12}}|$, we can compute that $|\overline{S_{N}}| = \frac{1}{2} n(n-s_{1})p_{n}$, and  $|\overline{S_{1N}}| = \frac{1}{2}n(s_{1}-1)p_{1n}$ where $p_{1n}=p_1+p_n-p_1p_n$.

\end{proof}

\begin{theorem}\label{(theorem)modularity of model(2) > modularity of model(1)}
For the ground truth $l_1$, its modularity in \Gstruct is less than its corresponding modularity in \Grandom, that is,  $\Mstruct_{l_1}< \Mrandom_{l_1}$.

\end{theorem}

\begin{proof}
	By Lemma~\ref{(lemma)ell_and_elout_2_layer_model},
we can reduce the inequality (\ref{(equation)jialu2020_in_why_use_hicode}) to :
$\frac{|\overline{S_{1}}|+|\overline{S_{12}}|}{|\overline{S_{2}}|} < \frac{|\overline{S_{1N}}|}{|\overline{S_{N}}|}$.
By Lemma~\ref{(lemma)(why use hicode)The maximum number of edges in different set}, this inequality holds iff
\begin{align*}
& \frac{(s_{1}-r_{12})p_{1}+(r_{12}-1)p_{12}}{(s_{2}-r_{12})p_{2}}  < \frac{(s_{1}-1)p_{1n}}{(n-s_{1})p_{n}},
\end{align*}
which by simple algebraic manipulation could be further reduced to
\begin{align}
	\label{Theorem 8: inequality}
	\frac{p_1}{p_2} \cdot (s_1 -1) + (r_{12} - 1) \cdot (1-p_1)
	< \frac{1}{n_2} \left( \frac{p_1}{p_n} \cdot (s_1 -1) + (\frac{n}{n_1} -1) \cdot (1 - p_1)\right).
\end{align}
Recall that we set $p_n = \frac{s_2 -1}{n-1} \cdot p_2$, so $p_n < \frac{p_2}{n_2}$, and thus the right hand side is greater than
$ \frac{p_1}{p_2} \cdot (s_1 -1) + \frac{1}{n_2} (\frac{n}{n_1} -1) \cdot (1 - p_1)$.
Since $n_2 > 1$ and $r_{12} = \frac{n}{n_1 \cdot n_2}$,
~\cref{Theorem 8: inequality} always holds.

\end{proof}

\subsection{Structured Noise Cause Two Peaks}\label{(sec)Struct_Cause_P}
We then show that the modularity curve surface of \Gstruct has at least two peaks, indicating that the modularity optimizing community detection algorithms may fall into a local optimum. 
We demonstrate that for any partition very close to $l_2$ ($l_1$), its modularity is less than the modularity of $l_2$ ($l_1$).  We define ``close partitions'' as follows:


\begin{definition} (Close Partitions.)
	Partitions $l'$ is close to partition $l$ if $l'$ can be obtained from $l$ through one of the following updates:

1) move one node from a community to another community;
2) exchange two nodes' community membership;
3) separate one node from a community to form a new community. 
\end{definition}
In the following, we show that both layers in \Gstruct can give rise to modularity peaks
as long as the layer has enough structure.
For simplicity, we focus on $l_2$; the proof for $l_1$ is analogous.

\begin{theorem}\label{(the)two_p}
In \Gstruct, if the ground truth partition $l_2$ has three times more internal edges than outgoing edges, that is,  $3e_{22}>e_{2out}$,
then the modularity of $l_2$ is higher than the modularity of any
$l'_2$ close to $l_2$, i.e.,  $Q^{\textsf{S}}_{l_2}>Q^{\textsf{S}}_{l'_2}$.
\end{theorem}

\begin{proof}
Intuitively, the restriction $3e_{22}>e_{2out}$ makes sure that communities in $l_2$
have significantly stronger internal connections.
We will only show the proof for partitions obtained through the first case because the other two cases can be proved similarly. See appendix for proofs for the other two cases.

For the first case, without loss of generality,
we assume that node $v$ in community 1 is moved to community 2.
Before the update, there expect to be $\frac{e_{22}}{s_{2}}$ internal edges connected with node $v$. Also, there expect to be $\frac{e_{2out}}{(n_{2}-1)s_{2}}$ outgoing edges from node $v$ to community 2, and $\frac{(n_{2}-2)e_{2out}}{(n_{2}-1)s_{2}}$ outgoing edges from $v$ to communities other than community 2.

After moving, the internal edges connected to node $v$ become the outgoing edges of both community 1 and 2. The outgoing edges from node $v$ to community 2 and other communities respectively become the internal and outgoing edges of community 2.

Therefore, in the new partition, for community 1, $e'_{22}=e_{22}-\frac{e_{22}}{s_{2}}$ and $e'_{2out}=e_{2out} - \frac{e_{2out}}{s_{2}}$. For community 2, $e'_{22}=e_{22}+\frac{e_{2out}}{(n_{2}-1)s_{2}}$ and $e'_{2out}=e_{2out} + \frac{e_{22}}{s_{2}}+ \frac{(n_2-2)e_{2out}}{(n_2-1)s_{2}}$. The modularity of the new partitions:

\begin{align*}
Q^{\textsf{S}}_{l'_2}=&\sum_{i=1}^{n_{2}} \frac{e_{22}^{i'}}{e} - \left(\frac{d_{2}^{i'}}{2e} \right)^2 \\
=& Q^{\textsf{S}}_{l_2} - \frac{e_{22}}{e\cdot s_{2}} + \frac{e_{2out}}{e \cdot s_{2}(n_{2}-1)}- \frac{\left[ d_{2}^1-\left( \frac{2e_{22}}{s_{2}}+\frac{e_{2out}}{s_{2}} \right) \right]^2}{4e^2}+ \frac{d_{2}^2}{4e^2} \\
&- \frac{\left[ d_{2}^2 + \left( \frac{n_{2}e_{2out}}{(n_2-1)s_{2}}+\frac{e_{22}}{s_{2}} \right) \right]^2}{4e^2}  + \frac{d_{2}^2}{4e^2}\\
<&  Q^{\textsf{S}}_{l_2} - \left[ \frac{(n_2-1)e_{22}-e_{2out}}{e\cdot s_{2}(n_2-1)} - \frac{2d_{2}^1\left(\frac{2e_{22}}{s_{2}}+\frac{e_{2out}}{s_{2}} \right)}{4e^2} + \frac{2d_{2}^2\left(\frac{e_{22}}{s_{2}}+\frac{n_{2}e_{2out}}{(n_{2}-1)s_{2}} \right)}{4e^2}  \right] \\
=& Q^{\textsf{S}}_{l_2} - \left[ \frac{(n_2-1)e_{22}-e_{2out}}{e\cdot s_{2}(n_2-1)} - \frac{2e_{22}+e_{2out}}{e\cdot s_{2}n_{2}}+\frac{e_{22}+\frac{n_{2}}{n_{2}-1}e_{2out}}{e \cdot s_{2}n_{2}} \right]\\
= & Q^{\textsf{S}}_{l_2} -  \frac{1}{e\cdot s_{2}n_{2}(n_{2}-1)}[n_{2}(n_{2}-1)e_{22}-(n_{2}-1)e_{22}-(n_{2}-1)e_{2out}] \\
= &  Q^{\textsf{S}}_{l_2} - \frac{1}{ne}[(n_{2}-1)e_{22}-e_{2out}].
\end{align*}
Because $ 3e_{22}>e_{2out} $ and $ n_{i}>4 $, $(n_{2}-1)e_{22}-e_{2out}>0$ and $Q^{\textsf{S}}_{l'_2}< Q^{\textsf{S}}_{l'_2}$ holds.

\end{proof}
Thus, there are at least two peaks for the modularities of \Gstruct, which obstructs the detection of the dominant structure.

\section{Theoretical Analysis on Three-layer Stochastic Block Model}\label{section(Theoretical analysis on three-layer Stochastic Block Model)}
Next, we give some evidence on why HICODE works when there are more than two layers. Because the proof for the general multi-layer case uses complicated notations to bookkeep the layers, we first walk through the three-layer case in this section to show the gist of our proof. In particular, we prove that, after performing RemoveEdge / ReduceEdge / ReduceWeight on the three-layer stochastic block model network, the modularity of the target layer always increases, making it easier for the base algorithm to uncover the target layers.

We first introduce some preliminary 
lemmas.

\begin{lemma}\label{ell_and_elout}
In the three-layer stochastic block model $G(n, n_1, n_2,$ $ n_3, p_1, p_2, p_3 )$, the number of internal edges and outgoing edges of one community in layer 1, layer 2 and layer 3 are:
\begin{align*}
e_{11}=\frac{1}{n_1}(|\overline{S_1}|+|\overline{S_{12}}|+|\overline{S_{13}}|+|\overline{S_{123}}|), \qquad e_{1out}=\frac{2}{n_1}(|\overline{S_{2}}|+|\overline{S_{3}}|+|\overline{S_{23}}|),\\
e_{22}=\frac{1}{n_2}(|\overline{S_2}|+|\overline{S_{12}}|+|\overline{S_{23}}|+|\overline{S_{123}}|), \qquad
e_{2out}=\frac{2}{n_2}(|\overline{S_{1}}|+|\overline{S_{3}}|+|\overline{S_{13}}|),\\
e_{33}=\frac{1}{n_3}(|\overline{S_3}|+|\overline{S_{13}}|+|\overline{S_{23}}|+|\overline{S_{123}}|), \qquad
e_{3out}=\frac{2}{n_3}(|\overline{S_{1}}|+|\overline{S_{2}}|+|\overline{S_{12}}|).\\
\end{align*}
\end{lemma}

\begin{lemma}\label{The maximum number of edges in different sets}
The expected number of edges in different edge sets on model $G(n, n_1, n_2,$ $ n_3, p_1, p_2, p_3 )$ are:
\begin{align*}
|\overline{S_{123}}| &= \frac{1}{2}n(r_{123}-1)p_{123},&
|\overline{S_{12}}| &= \frac{1}{2}n(n_3-1)r_{123}p_{12},\\
|\overline{S_{13}}| &= \frac{1}{2}n(n_2-1)r_{123}p_{13},&
|\overline{S_{23}}| &= \frac{1}{2}n(n_1-1)r_{123}p_{23},\\
|\overline{S_{1}}| &= \frac{1}{2}n(n_2-1)(n_3-1)r_{123}p_{1},&
|\overline{S_{2}}| &= \frac{1}{2}n(n_1-1)(n_3-1)r_{123}p_{2},\\
|\overline{S_{3}}| &= \frac{1}{2}n(n_1-1)(n_2-1)r_{123}p_{3}.
\end{align*}
\end{lemma}

The proof of Lemma \ref{ell_and_elout} and \ref{The maximum number of edges in different sets} are similar to Lemma  \ref{(lemma)(why use hicode)The maximum number of edges in different  set} and \ref{(lemma)ell_and_elout_2_layer_model}, so we omit them here.

Now we return to the main points of this section. HICODE is an iterative process, and at each iteration, it takes a layer as the target layer and reduces other layers iteratively. Our intuition is that if we can reduce one layer, all other layers will be more visible in the network. Then, HICODE will locate each layer more accurately as the iterative process continues and eventually approximate the ground truth layers very closely. Concretely, we prove that:


\begin{theorem}
For a three-layer stochastic block model network $G(n,n_1,n_2,n_3,p_1,p_2,p_3)$, the modularity of a layer increases if we apply \textcolor{blue}{RemoveEdge} on all communities in other layers.
\end{theorem}
\begin{proof}
We consider one case of taking layer 2 as the target layer and applying the reducing method on the other two layers. Taking layer 1 or layer 3 as the target layer can be proved similarly. 

After we apply RemoveEdge, all of the edges generated by layer 1 and layer 3 are removed. Therefore, the values of $|\overline{S_1}|$, $|\overline{S_3}|$, $|\overline{S_{12}}|$, $|\overline{S_{13}}|$, $|\overline{S_{23}}|$ and $|\overline{S_{123}}|$ 
become to 0, and $|\overline{S_2}|$ is unaffected.
Therefore, we have:
\begin{align*}
	\frac{e'_{22}}{e_{22}} &= \frac{(|\overline{S_2}|'+|\overline{S_{12}}|'+|\overline{S_{23}}|'+|\overline{S_{123}}|')}{(|\overline{S_2}|+|\overline{S_{12}}|+|\overline{S_{23}}|+|\overline{S_{123}}|)} \\
																								&>\frac{(|\overline{S_1}|'+|\overline{S_3}|'+|\overline{S_{13}}|')}{(|\overline{S_1}|+|\overline{S_3}|+|\overline{S_{13}}|)}
																								= \frac{e'_{2out}}{e_{2out}}
\end{align*}
By Lemma~\ref{jialu2020}, we could complete the proof.
\end{proof}

\begin{theorem}\label{(3layer)ReduceEdge on other layer}
For a three-layer stochastic block model network $G(n,n_1,n_2,n_3,p_1,p_2,p_3)$, the modularity of a layer increases if we apply \textcolor{blue}{ReduceEdge} on all communities in other layers.
\end{theorem}
\begin{proof}
We still consider one case of taking layer 2 as the target layer. 
ReduceEdge respectively removes the edges in layer 1 and layer 3 with probability $1-q_{1}$ and $1-q_{3}$, where $q_i$ is defined in Eq.~\ref{defn_qi}. The edges in different sets are affected differently. When reducing layer 1, $S_{1}$, $S_{12}$, $S_{13}$ and $S_{123}$ keep edges with ratio $q_{1}$. When reducing layer 3, $S_{3}$, $S_{13}$, $S_{23}$ and $S_{123}$ keep edges with ratio $q_{3}$. 

Thus, the corresponding expectations are $|\overline{S_{3}}|q_{3}$, $|\overline{S_{23}}|q_{3}$, $|\overline{S_{13}}|q_{1}q_{3}$ and $|\overline{S_{123}}|q_{1}q_{3}$ after reducing layer 1 and layer 3. Hence we have:
\begin{align*}
&\frac{e'_{22}}{e_{22}} >\frac{e'_{2 out}}{e_{2out}}\\
\iff &\frac{|\overline{S_{2}}|+|\overline{S_{12}}|q_{1}+|\overline{S_{23}}|q_{3}+|\overline{S_{123}}|q_{1}q_{3}}{|\overline{S_{2}}|+|\overline{S_{12}}|+|\overline{S_{23}}|+|\overline{S_{123}}|} > \frac{|\overline{S_{1}}|q_{1}+|\overline{S_{3}}|q_{3}+|\overline{S_{13}}|q_{1}q_{3}}{|\overline{S_{1}}|+|\overline{S_{3}}|+|\overline{S_{13}}|}\\
\iff & (1-q_{1})\cdot [|\overline{S_{1}}||\overline{S_{2}}|+|\overline{S_{1}}||\overline{S_{23}}|-|\overline{S_{3}}||\overline{S_{12}}|+(|\overline{S_{23}}||\overline{S_{13}}|-|\overline{S_{3}}||\overline{S_{123}}|)q_{3}]\\
+ & (1-q_{3}) \cdot [|\overline{S_{2}}||\overline{S_{3}}|+|\overline{S_{3}}||\overline{S_{12}}|-|\overline{S_{1}}||\overline{S_{23}}|+(|\overline{S_{12}}||\overline{S_{13}}|-|\overline{S_{1}}||\overline{S_{123}}|)q_{1}]\\
+ &(1-q_{1}q_{3})|\overline{S_{2}}||\overline{S_{13}}| > 0. 
\end{align*}
If we apply Lemma \ref{ell_and_elout} and Lemma \ref{The maximum number of edges in different sets} and eliminate $\frac{1}{4}n^2(n_{1}-1)(n_{2}-1)(n_{3}-1)r_{123}^2$, we can simplify the inequality into:
\begin{align*}
&(1-q_{1}) \cdot [(n_{3}-1)p_{1}p_{2}+p_{1}p_{23}-p_{3}p_{12}+(p_{23}p_{13}-p_{3}p_{123})] \\
+& (1-q_{3}) \cdot [(n_{1}-1)p_{2}p_{3}+p_{3}p_{12}-p_{1}p_{23}+(p_{12}p_{13}-p_{1}p_{123})]\\
+&(1-q_{1}q_{3}) \cdot p_{2}p_{13} > 0 \\
\end{align*}
When all $n_i \geq 4$, it suffices if
\begin{align*}
 & (1-q_{1})(3p_{1}p_{2}+p_{1}p_{23}-p_{3}p_{12}) + (1-q_{3})(3p_{2}p_{3}+p_{3}p_{12}-p_{1}p_{23})>0, \quad
\end{align*}
We can further simplify it into
\begin{align*}
& p_{1}p_{2}(3+q_{3}-4q_{1})+p_{2}p_{3}(3+q_{1}-4q_{3})>0.
\end{align*}
Because of $q_i \leq0.75$, this inequality holds, and thus $\frac{e'_{22}}{e_{22}} >\frac{e'_{2 out}}{e_{2out}}$ holds.
\end{proof}

Though we assume a restriction on $n_i$ and $q$'s values in the proof, $q_i$ can be bigger if $n_i$ is guaranteed to be bigger. This fact indicates that even when the edge density in the background is approaching the internal edge density in the target layer, if the number of communities is guaranteed to be big, HICODE can still work. It makes sense because when there are many communities, each community is smaller,  and each node has less neighbors from the same community and more neighbors from other communities, and thus it is harder to have significantly higher internal edge density as compared to the background density.

\begin{theorem}\label{(3layer)ReduceWeight on other layer}
For a three-layer stochastic block model network $G(n,n_1,n_2,n_3,p_1,p_2,p_3)$, the modularity of a layer increases if we apply \textcolor{blue}{ReduceWeight} on all communities in other layers.
\end{theorem}
\begin{proof}
We still consider one case of taking layer 2 as the target layer. 
In the weighted network, the weighted sum of internal edges and outgoing edges of a community $i$ in layer 2 is $e_{22} = \frac{1}{2} \sum_{u,v \in i} w_{uv} A_{uv}$ and $e_{2out} = \frac{1}{2} \sum_{u \in i, v \notin i} w_{uv} A_{uv}$ where $w_{uv}$ is the weight of edge $(u,v)$ and $A_{uv}$ is an indicative function. 
In HICODE, ReduceWeight is parallel to ReduceEdge. Therefore, the internal and outgoing edges before and after the reducing are:
\begin{align*}
e_{22} &=\frac{1}{2n_2}\sum_{\{{S_2},S_{12},S_{23},S_{123}\}}^{S_i}\sum_{(u,v) \in {S_i}} w_{uv} \cdot A_{uv},\\
e_{22}' &= \frac{1}{2n_2}(\sum_{(u,v) \in {S_2}} w_{uv} \cdot A_{uv} + \sum_{(u,v) \in {S_{12}}} w_{uv}\cdot q_{1} \cdot A_{uv} \\
&+ \sum_{(u,v) \in {S_{23}}} w_{uv}\cdot q_{3} \cdot A_{uv}+ \sum_{(u,v) \in {S_{123}}} w_{uv}\cdot q_{1}q_{3} \cdot A_{uv} ,\\
e_{2out} &=\frac{1}{2n_2}\sum_{\{{S_1},S_{3},S_{13}\}}^{S_i}\sum_{(u,v) \in {S_i}} w_{uv}\cdot A_{uv}),\\
e_{2out}' &= \frac{1}{2n_2}(\sum_{(u,v) \in {S_{1}}} w_{uv}\cdot q_{1} \cdot A_{uv} + \sum_{(u,v) \in {S_{3}}} w_{uv}\cdot q_{3} \cdot A_{uv} + \sum_{(u,v) \in {S_{13}}} w_{uv}\cdot q_{1}q_{3} \cdot A_{uv}  ).\\
\end{align*}
In our three-layer stochastic block model, $w_{uv}$ is the generation probability of edge $(u,v)$. Therefore, after some algebraic manipulation, we can get the same inequality as in Theorem \ref{(3layer)ReduceEdge on other layer}.
\end{proof}

The analysis shows that when HICODE applies the reducing method on other layers, the modularity of the current target layer always increases, implying that the ground truth communities in the target layer 
are getting stronger in the reduced graph 
and become easier to be detected.

\section{Theoretical Analysis on Multi-layer Stochastic Block Model}\label{section(Theoretical analysis on Multi-layer Stochastic Block Model)}
This section analyzes the HICODE's effect on the strength of layers in the general multi-layer stochastic block model. 
Similarly, we prove that fixing a target layer and applying the reduction method on all other layers would increase the target layer's modularity. The proofs of different reduction methods (RemoveEdge, ReduceEdge, and ReduceWeight) are analogous. Hence, we unify them into one proof.

\begin{theorem}\label{(The)reduce_in_mult_SBM}
For a multi-layer stochastic block model network $G(n,n_1...n_L,p_1...p_L)$.  Let $\mathcal{L}$ denote the set of all layers. For any target layer $l^* \in \mathcal{L}$, let $\mathcal{L'} = \mathcal{L} \setminus \{l^*\}$. If we apply any ``suitable'' reducing method on all communities in  $\mathcal{L'}$, the modularity of $l^*$ always increases.
\end{theorem}

We will define ``suitable'' shortly, but first, let us define several functions to simplify the notations.

\begin{definition}
Define a function $ C : 2^{\mathcal{L'}} \rightarrow 2^{\mathcal{L'}}$ to return the complement of a set with respect to $\mathcal{L'}$, that is,  $C(T) = \mathcal{L'} \setminus T $.
\end{definition}
\begin{definition}
Define a function $ F: 2^{\mathcal{L'}} \rightarrow \mathbb{Q}$ to return the fractions of node pairs that are internal to 
all layers in $T$ and not internal to any layer in $C(T)$.
\end{definition}

Since the probability that a node pair is internal to layer $l$ is $\frac{1}{n_l}$ (here we assume that nodes can have self-loops for convenience), and layers are independent,
\[F(T) = \prod_{l \in T} \frac{1}{n_l} \cdot \prod_{l \in C(T)} (1-\frac{1}{n_l}).\]


\begin{definition}
Define a function $P: 2^{\mathcal{L}} \rightarrow [0,1]$ such that $P(T)$ is the probability that one of the layers in $T$ generates an edge on a node pair internal to layers $T$. 
\end{definition}
Because all layers are independent, we can define $P$ recursively: 
for $l \in \mathcal{L}$ such that $l \notin T$, 
\begin{align*}
P(\{l\} \cup T) = P(\{l\}) + ( 1 - P(\{l\})) \cdot P(T)
\end{align*}

\begin{definition}
Define a function $K: 2^{\mathcal{L}'} \rightarrow [0,1]$ such that $K(T)$ is the probability that the reduction method keeps edges in $T$. 
\end{definition}

We consider a reducing method to be ``suitable'' if for any $T, T' \in 2^{\mathcal{L}'}$, $P(T) \leq P(T')$ implies $K(T) \geq K(T')$. 
When HICODE applies RemoveEdge, then $K(T) = 0$ for any set of layers $T$, so RemoveEdge is suitable. 
When HICODE applies ReduceEdge or ReduceWeight, we reduce edges internal to $T$ so that the edge density in $T$ is the same as the background edge density. 
Also, a larger $P(T)$ indicates that an edge internal to $T$ is more likely to be generated, so intuitively, 
so more percentage of edges in the set will be removed to match the background. 
Formally, We conjecture that ReduceEdge and ReduceWeight are suitable under some restrictions on the number of communities and background edge density, like the case in~\cref{(3layer)ReduceEdge on other layer} and~\cref{(3layer)ReduceWeight on other layer}, 

\begin{proof}[Theorem\ref{(The)reduce_in_mult_SBM}]
	For any node pair in any $T \in 2^{\mathcal{L'}}$, $F(T) \cdot P(T)$ represents the probability that this node pair generates an edge which is internal to layers in $T$ and outgoing to layers in $C(T)$. Fix $l^*$ as the target layer. There are a total of $\frac{1}{2}\frac{n}{n_{l^*}}(\frac{n}{n_{l^*}}-1)$ and $\frac{n}{n_{l^*}} \cdot (n-\frac{n}{n_{l^*}})$ node pairs that can generate internal and outgoing edges for layer $l^*$. Also, since layers are independent, the number of node pairs in $\{l^*\} \cup  T$ (and no other layers) is the number of node pairs in $l^*$ times $F(T)$.
	So we have:
\begin{align*}
e_{l^{*}l^{*}} &=  \frac{1}{2}\frac{n}{n_{l^*}}(\frac{n}{n_{l^*}}-1) \cdot \sum_{T \in 2^{\mathcal{L'}}}  F(T) \cdot P\big(\{l^*\} \cup T\big) , \\
e_{l^{*}out} &= \frac{n}{n_{l^*}}(n-\frac{n}{n_{l^*}}) \cdot \sum_{T \in 2^{\mathcal{L'}}} F(T) \cdot P(T).
\end{align*}

After the reducing method is performed on all communities in  $\mathcal{L'}$,
\begin{align*}
e_{l^{*}l^{*}}' &=\frac{1}{2}\frac{n}{n_{l^*}}(\frac{n}{n_{l^*}}-1) \cdot \sum_{T \in 2^{\mathcal{L'}}}  F(T) \cdot P\big(\{l^*\} \cup T\big) \cdot K(T), \\
e_{l^{*}out}' &= \frac{n}{n_{l^*}} (n-\frac{n}{n_{l^*}}) \cdot \sum_{T \in 2^{\mathcal{L'}}} F(T) \cdot P(T) \cdot K(T).
\end{align*}
Thus,
\begin{align*}
&\frac{e'_{l^*out}}{e_{l^*out}} < \frac{e'_{l^*l^*}}{e_{l^*l^*}} \\
\iff & \frac{\sum_{T \in 2^{\mathcal{L'}}} F(T) \cdot P(T) \cdot K(T)}{\sum_{T \in 2^{\mathcal{L'}}} F(T) \cdot P(T)} <  \frac{\sum_{T \in 2^{\mathcal{L'}}} F(T) \cdot P\big(\{l^*\} \cup T\big) \cdot K(T) }{\sum_{T \in 2^{\mathcal{L'}}}  F(T) \cdot P\big(\{l^*\} \cup T\big) }
\end{align*}
Since $P(\{l^*\} \cup T) = P(\{l^*\}) + ( 1 - P(\{l^*\})) \cdot P(T) $, so the inequality is equivalent to
\begin{align*}
	& [\sum_{T \in 2^{\mathcal{L'}}} F(T)P(T)K(T) \sum_{T \in 2^{\mathcal{L'}}} F(T)P(T))] \cdot (1-P(l^{*}))\\
+ & \sum_{T \in 2^{\mathcal{L'}}} F(T)P(T)K(T)  \sum_{T \in 2^{\mathcal{L'}}}F(T)P(l^{*}) \\
> & [\sum_{T \in 2^{\mathcal{L'}}} F(T)P(T) \sum_{T \in 2^{\mathcal{L'}}} F(T)P(T) K(T)] \cdot (1-P(l^{*})) \\
+ &\sum_{T \in 2^{\mathcal{L'}}}F(T)P(T)\sum_{T \in 2^{\mathcal{L'}}}F(T)P(l^{*}) K(T),
\end{align*}
which is equivalent to
\begin{align*}
\iff & \sum_{T \in 2^{\mathcal{L'}}}F(T)P(T)\sum_{T \in 2^{\mathcal{L'}}} F(T)K(T) - \sum_{T \in 2^{\mathcal{L'}}}F(T)P(T)K(T)\sum_{T \in 2^{\mathcal{L'}}}F(T) > 0.
\end{align*}
Represent the elements in the power set of $\mathcal{L'}$ as $T^0, T^1... T^{2^{|L'|}}$, and sort the elements such that for any $i,j \in [0,2^{|L'|}]$, if $i<j$ , $P(T^i) \leq P(T^j)$ and $K(T^i) \geq K(T^j)$.
Therefore, the above inequality is equivalent to:
\begin{align}\label{multi:inequality}
\sum_{i=0}^{2^{|L'|}}\sum_{j=0}^{2^{|L'|}}F(T^{i})F(T^{j})K(T^{j})(P(T^{i})-P(T^{j})) > 0.
\end{align}
For any $i,j \in [0, 2^{|L'|}]$ and $i < j$, we have:
\begin{align*}
F(T^{i})F(T^{j})K(T^{j})(P(T^{i})-P(T^{j})) +  F(T^{j})F(T^{i})K(T^{i})(P(T^{j})-P(T^{i})) > 0.
\end{align*}
Thus, ~\cref{multi:inequality} holds and we have $\frac{e'_{l^*out}}{e_{l^*out}} < \frac{e'_{l^*l^*}}{e_{l^*l^*}}$.
\end{proof}

\section{Simulation Results of HICODE}\label{(section)Simulation of HICODE}
In this section, we simulate the process of how the community structure could be strengthened through the iterative reduction process of HICODE. 
We adapt normalized mutual information (NMI)~\citep{danon2005comparing} to measure the similarity between two partitions. Intuitively, two partitions are more similar if they have a larger NMI value. 

\begin{definition}[NMI similarity]\label{(definition)NMI similarity}
The normalized mutual information (NMI) of two partitions $X,Y$ is defined to be
\begin{align*}
NMI(X,Y) = \frac{2 I(X,Y)}{ H(X) + H(Y)},
\end{align*}
where $H(X)$ is the entropy of partition with $p(x)$ taken to be $|X|$.
\begin{align*}
H(X) = - \sum_{x \in X} p(x) \log p(x)
= - \sum_{x \in X} |x| \log |x|,
\end{align*}
and $I(X,Y)$ measures the mutual information between $X$ and $Y$ by
\begin{align*}
I(X,Y) &= \sum_{x \in X} \sum_{y \in Y} p(x,y) \log \frac{p(x,y)}{p(x) \cdot p(y)}\\
&= \sum_{x \in X} \sum_{y \in Y} |x \cap y| \log \frac{ |x \cap y| }{|x| \cdot |y|}.
\end{align*}
\end{definition}

To illustrate the effectiveness of HICODE on multi-layer networks more vividly, we run HICODE on a network synthesized from 3-layer stochastic block model.
Specifically, we generate the network based on $G(400,4,5,10,0.1,0.11,0.12)$ and adopt ReduceEdge as the reducing method and Louvain as the base algorithm.
We visualize a set of possible partitions' modularity at different timestamps of the algorithm. We also highlight the partition found by the base algorithm at each step to show how it gets close to the ground truth layers. 

Ideally, we want to enumerate all possible partitions of $n$ nodes and show their modularity changes, but it is computationally unrealistic because the number of possible partitions is exponential. Thus, we sample possible partitions in two ways. The first method is by mutating the ground truth layers. Specifically, we start from a ground truth layer,
and exchange $k$ random pairs of nodes for $k = 1, ..., 500$. We collect two samples for each $k$ and 1000 samples in total through this method. The second method is through blending all the layers: it generates a ``blended'' partition by letting each node's community ID be its community ID in one of the ground truth layers. 
In these ways, we sample at least 4500 partitions. 

In~\cref{(figure)3layersimulation}, we present the partitions' modularity at each timestamp in a similar manner to \cref{(figure)modularity_of_two_model}. 
In each sub-figure of \cref{(figure)3layersimulation}, a partition's $x$ and $y$ coordinates indicate its similarity to the two ground truth layers, where the similarity is measured by the NMI scores, 
and the $z$ coordinate indicates its modularity. 
The blue grid shows the sampled partitions' modularity and the red vertical line is the partition found by the base algorithm. 
We choose two layers instead of using all three ground truth layers when determining the coordinates because otherwise, the data would be 4-dimensional and hard to visualize; 
to compensate, we use three sub-figures to display three perspectives resulted from different choices of layers: layer 1 and 2 for the first column,
layer 1 and 3  for the second column, and layer 2 and 3 for the third column.

\begin{figure}[htb]
\center{\includegraphics[width=16cm] {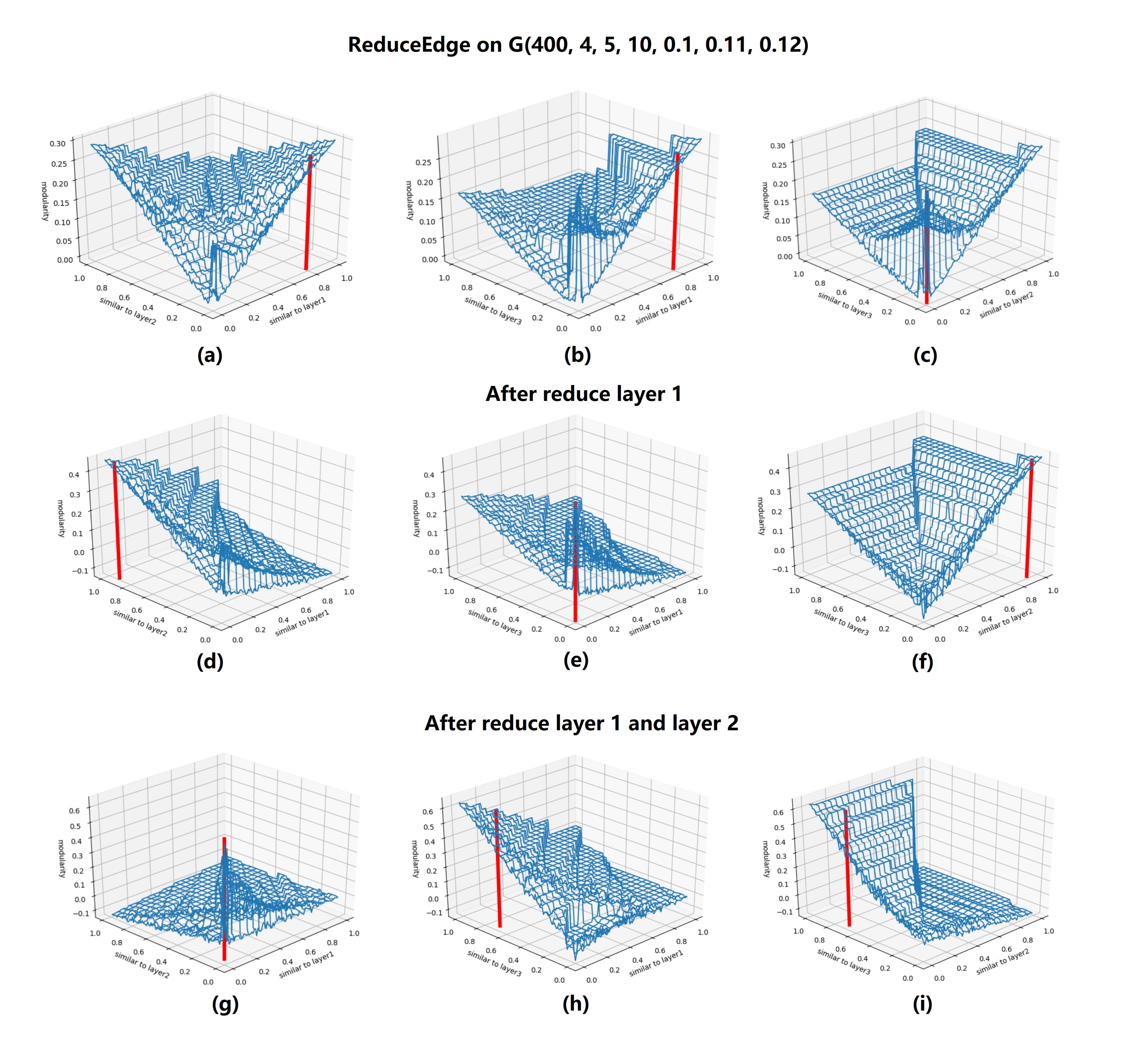}}
\caption{\label{1} The identification stage of HICODE on a three-layer stochastic model.}
\label{(figure)3layersimulation}
\end{figure}

Specifically, 
\begin{itemize}
\item On the first row, the sub-figures (a), (b), and (c) represent the situation on the original graph $G_0$. The modularity scores of the three ground truth layers are respectively 0.300, 0.292, and 0.166. The base algorithm finds a partition whose NMI with layer 1 is 0.79, so it almost uncovers layer 1 and overlooks layers 2 and 3. 

\item On the second row, the sub-figures (d), (e), and (f) represent the situation on graph $G_1$ after we reduce the detected layer 1 on $G_0$.
We observe that the modularity of layer 1 drops significantly, and the modularity scores of layers 2 and 3 increase to 0.458 and 0.279 respectively. 
As the result, the base algorithm detects a partition whose $NMI$ with layer 2 is $0.90$. 

\item On the third row, the sub-figures (g), (h), and (i) represent the situation on graph $G_2$, which is obtained from reducing the detected layer 2 on $G_1$. Here, both the modularity scores of layer 1 and layer 2 have decayed significantly. In contrast, the modularity of layer 3 increases to be as high as 0.645, making it much easier for the base algorithm to uncover. 
\end{itemize}


\begin{figure}[htb]
\center{\includegraphics[width=13cm]  {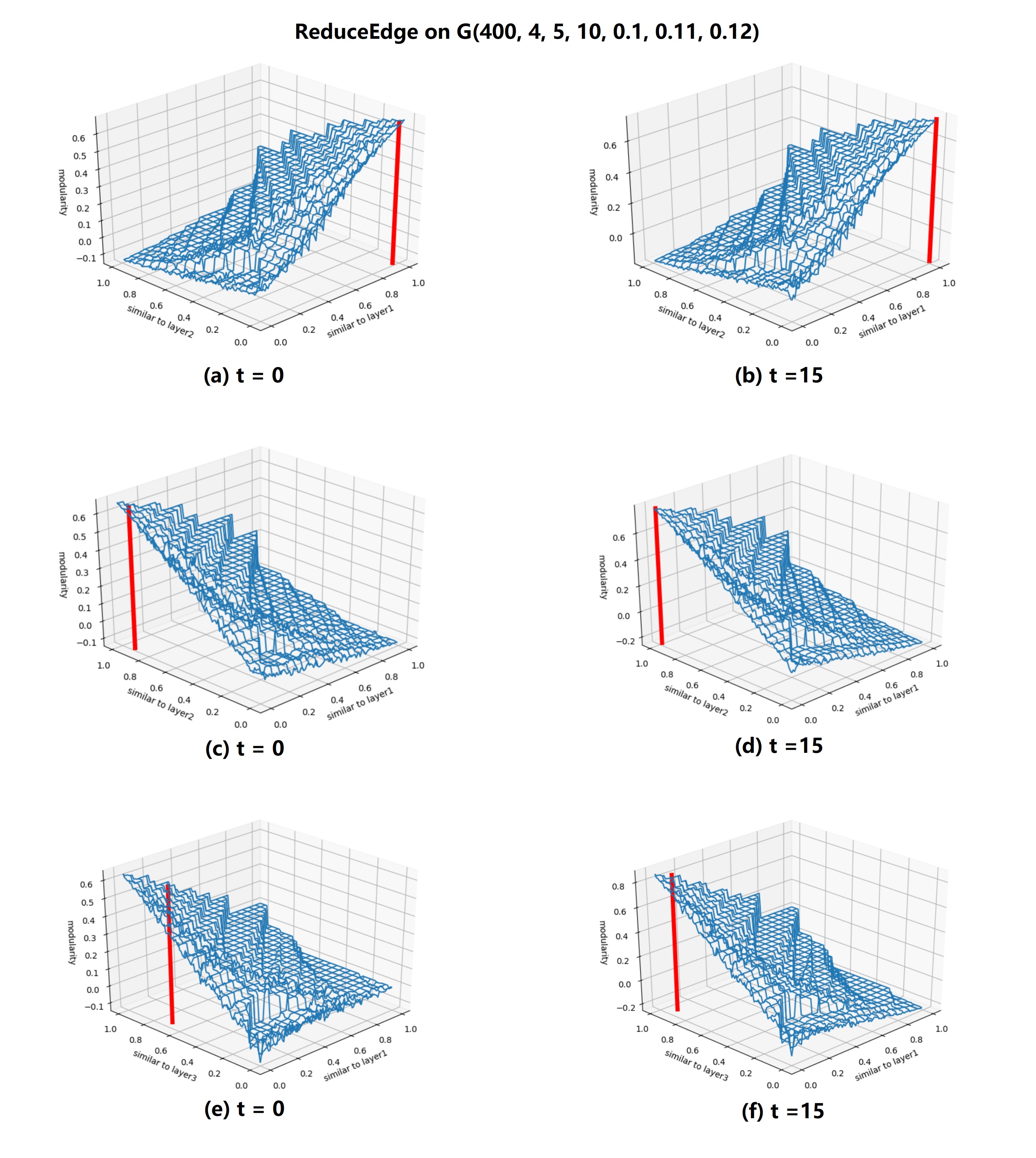}}
\caption{\label{2} The refinement stage of HICODE on the three-layer stochastic model.}
\end{figure}

Fig. \ref{2} shows the effectiveness of the refinement stage. The red vertical lines in sub-figures (a), (c), and (e) represent the partitions detected in the identification stage to approximate different layers, and the red vertical lines in sub-figures (b),(d), and (f) represent the partitions detected after 15 iterations of the refinement stage. The red lines in (b), (d), and (f) are closer to the corners of the bottom square,  indicating that each detected partition approximates the ground truth layer it tries to uncover better after the refinement stage. In the end, the NMI of the ground truth layers and respective detected partitions are 1.00, 0.98, and 0.91, demonstrating the boosting effects of the refinement stage. 

\section{Conclusion}\label{section(Conclusion)}

Real-world networks may contain multiple layers of community structure, where each layer consists of mostly disjoint communities. In that case, the stronger layers could dominate the structure and hide the weaker layers. Typical community detection algorithms overlook the hidden communities overshadowed by the dominant communities. Moreover, the existence of hidden communities also interferes with the detection of the dominant communities.

We model real-world networks through multi-layer stochastic block models and analyze the effects of hidden communities on community detection tasks in a systematic manner. We first distinguish random noise added to the dominant communities from the equivalent structured noise coming from the hidden communities. We demonstrate that the structured noise will generate another modularity peak for node partitions and hinders the detection of the dominant communities more,  comparing to the commensurate random noise. 

Then we provide theoretical support for the HICODE framework, which uses iterative reduction to separate the layers of community structures and to boost the detection quality. 
It had been proved that reducing more percentage of a layer's ``noise edges'' than its ``community edges'' would strengthen the communities in the layer. Based on that intuitive fact, we prove that HICODE's layer reduction process would always increase the target layers' modularity in multi-layer stochastic block models and improve the base algorithm's ability to discover the target layers. In the end, we also provide simulation on a three-layer stochastic block model, demonstrating visually that HICODE can preeminently improve the detection accuracy for all layers of community blocks.

\acks{This work is supported by National Natural Science Foundation (62076105).}


\newpage

\appendix
\section*{Appendix A.}
\label{app:theorem}


\begin{proof}[The other two cases in Theorem \ref{(the)two_p}]
For the second case, we swap node $v$ in community 1 and node $u$ in community 2. For these two nodes in the two communities, the internal edges of them become the outgoing edges of both communities, and the outgoing edges of them become the internal edges of each other’s community. 
Therefore, in the new partition, for both community, $e'_{22}=e_{22}-\frac{e_{22}}{s_{2}}+\frac{e_{2out}}{(n_{2}-1)s_{2}}$ and $e'_{2out}=e_{2out}+\frac{e_{22}}{s_{2}}-\frac{e_{2out}}{(n_{2}-1)s_{2}}$ .The modularity scores of the new partitions are as follows:

\begin{align*} 
Q^{\textsf{S}}_{l'_2}=&(n_{2}-2)\left[ \frac{e_{22}}{e} - \left(\frac{d_{2}}{2e} \right)^2  \right]+2\left[ \frac{e_{22}-\frac{e_{22}}{s_{2}}+\frac{e_{2out}}{(n_{2}-1)s_{2}}}{e} - \left(\frac{d_{2}- \frac{e_{22}}{s_{2}}+\frac{e_{2out}}{(n_{2}-1)s_{2}}}{2e} \right)^2  \right]\\
<& Q^{\textsf{S}}_{l_2} + 2\left[\frac{e_{22}}{e}- \frac{(n_{2}-1)e_{22}-e_{2out}}{e\cdot s_2(n_{2}-1)} - \left( \left(\frac{d_{2}}{2e} \right)^2 - \frac{2d_{1}\left( \frac{e_{22}}{s_{2}}-\frac{e_{2out}}{(n_{2}-1)s_{2}} \right)}{4e^2} \right) \right] \\
&- 2\left[ \frac{e_{22}}{e} - \left(\frac{d_{2}}{2e} \right)^2  \right] \\
=& Q^{\textsf{S}}_{l_2} - 2\left[  \frac{(n_{2}-1)e_{22}-e_{2out}}{e\cdot s_2(n_{2}-1)} - \frac{(n_{2}-1)e_{22}-e_{2out}}{e\cdot s_1n_{2}(n_{2}-1)} \right]\\
=& Q^{\textsf{S}}_{l_2} - \frac{2}{ne}[(n_{2}-1)e_{22}-e_{2out}].
\end{align*}

For the third case, we separate node $v$ from community 1 to form the new community $i$, the internal edges connected to node $v$ become the outgoing edges of both community 1 and community $i$ while the outgoing edges connected to node $v$ become the outgoing edges of community $i$. Therefore, in the new partition, for community 1, $e'_{22} = e_{22}-\frac{e_{22}}{s_{2}}$ and $e'_{2out}=e_{2out}+ \frac{e_{22}}{s_{2}} - \frac{e_{2out}}{s_{2}}$. And for the new community $i$, $e'_{2out}=e_{2out}+ \frac{e_{22}}{s_{2}} + \frac{e_{2out}}{s_{2}}$ while there are no internal edges. The modularity score of the new partition is:

\begin{align*} 
Q^{\textsf{S}}_{l'_2} 
=&(n_{2}-1)\left[ \frac{e_{22}}{e} - \left(\frac{d_{2}}{2e} \right)^2  \right] + \left[ \frac{e_{22}-\frac{e_{22}}{s_{2}}}{e} - \left(\frac{d_{2}-\frac{e_{22}}{s_{2}}-\frac{e_{2out}}{s_{2}}}{2e} \right)^2  \right] \\
&+\left[ \frac{0}{e}- \left(\frac{\frac{e_{22}}{s_{2}}+\frac{e_{2out}}{s_{2}}}{2e} \right)^2 \right] \\
=& Q^{\textsf{S}}_{l_2} - \left[ \frac{e_{22}}{e \cdot s_{2}} - \frac{2d_{2}\left(\frac{e_{22}}{s_{2}}+\frac{e_{2out}}{s_{2}}\right)}{4e^2} + \frac{2\left(\frac{e_{22}}{s_{2}}+\frac{e_{2out}}{s_{2}}\right)^2}{4e^2} \right] \\
< & Q^{\textsf{S}}_{l_2} - \left[\frac{e_{22}}{e \cdot s_{2}} - \frac{d_{2}(e_{22}+e_{2out})}{2e^2 \cdot s_{2}} \right] \\
= & Q^{\textsf{S}}_{l_2} - \frac{2}{ne}[(n_{2}-1)e_{22}-e_{2out}].
\end{align*}
Because $ 3e_{22}>e_{2out} $ and $ n_{i}>4 $, $(n_{2}-1)e_{22}-e_{2out}>0$ and thus $Q^{\textsf{S}}_{l'_2}< Q^{\textsf{S}}_{l'_2}$ holds.
\end{proof}

\newpage

\vskip 0.2in
\bibliography{main}

\end{document}